\documentclass[12pt]{amsart}
\usepackage{graphicx}
\usepackage{amsmath,amsthm}
\usepackage{amssymb}
\usepackage{caption}
\usepackage[latin1]{inputenc}
\usepackage{enumerate}
\usepackage{datetime}

\usepackage[margin=2cm]{geometry} 
\usepackage{relsize}

\usepackage[all]{xy}

\newtheorem{thm}{Theorem}[section]

\newtheorem{lema}[thm]{Lemma}
\newtheorem{prop}[thm]{Proposition}

\theoremstyle{definition}
\newtheorem{obs}[thm]{Remark}

\theoremstyle{plain}

\newtheorem*{thm*}{Teorema}
\newtheorem*{cor*}{Corolario}
\newtheorem*{lema*}{Lema}
\newtheorem*{prop*}{Proposición}
\theoremstyle{definition}
\newtheorem*{defi*}{Definición}
\newtheorem*{Ej*}{Ejemplo}
\newtheorem*{obs*}{Nota}
\newtheorem*{Obs*}{Notas}

\newcommand{\cC}{{\mathcal C}}
\numberwithin{equation}{section}

\def\vol{\textrm{vol}}
\def\diver{\textrm{div}}
\def\dis{\displaystyle}
\def\A{\mathcal{C}^\infty(M)}

\def\C{\mathcal{C}}

\def\idiota{{\text{\larger[1]{\emph\i}}}}

\def\bemol{{\boldsymbol{\flat}}}

\def\be{\begin{equation}
}

\def\ee{\end{equation}
}

\renewenvironment{abstract}
 {\small
  \begin{center}
  \bfseries \abstractname\vspace{-.5em}\vspace{0pt}
  \end{center}
  \vskip .3cm
  \list{}{
    \setlength{\leftmargin}{.5cm}%
    \setlength{\rightmargin}{\leftmargin}%
  }%
  \item\relax}
 {\endlist}

\begin{document}

\title[Flux-variational formulation of relativistic perfect fluids]
 {Flux-variational formulation of relativistic perfect fluids}

\author{R. J.  Alonso-Blanco and J.  Mu\~{n}oz-D{\'\i}az}

\address{Departamento de Matem\'{a}ticas, Universidad de Salamanca, Plaza de la Merced 1-4, E-37008 Salamanca,  Spain.}
\email{ricardo@usal.es, clint@usal.es}

\maketitle

\bigskip

\begin{abstract}
We give a variational formulation of perfect fluids on a general pseudoriemannian manifold by variating tangent fields according the flux produced by them. In this approach no constraints are needed. As a result, Euler and continuity equations are obtained quite directly.
\end{abstract}
\bigskip

\section{Introduction}

Different variational approaches to perfect fluids are available for a long time to the present (see, for example, \cite{Taub,Hawking,Schutz,Brown,Ootsuka1}).  Most of them use constraints or potentials. Indeed, in \cite{Schutz} is proven that constraints are required for (ordinary) variational formulations of perfect fluids. In this paper we offer a formulation that eliminates the need for such resources. It is based on the application of what we have called \emph{flux variations} (see below).

A steady fluid on a manifold is, firstly, described by a vector field: the field of its velocities; that includes non steady fluids if we add a time coordinate by considering an space-time manifold. So, the object we will deal with is  a vector field, this is to say, a section $u$ of the tangent bundle $TM\to M$, where $M$ is a smooth manifold. In order to consider variations of $u$, we take infinitesimal contact transformations preserving the geometrical structure of $TM$. If, in addition, $M$ is orientable and endowed with a pseudoriemannian metric $g$, we get the corresponding volume element $\vol$. In some way, the most important aspect of $u$ is the infinitesimal flux that it produces: $\idiota_{\dis u}\vol$, which is the inner contraction of $u$ with $\vol$. This is why is natural to consider the variation that a vector field  produces on the flux $\idiota_{\dis u}\vol$, instead of the variation caused on $u$ itself. On the other hand, the function of ``length'' $\rho(u_x):=\sqrt{|g(u_x,u_x)|}$ will be interpreted as the ``density'' of the fluid moved along $u$ (see, for instance, \cite{DiracRelatividad}, p. 50). Finally, we will consider a Lagrangian density $\Phi(\rho)\vol$, where $\Phi$ is an arbitrary smooth function, and will take the infinitesimal variations of its restriction $\Phi(\rho)\vol|_u$, according to how the flux of $u$ varies. The vanishing of these variations produces the equations of the fluid dynamics (Euler equation) along with the conservation of mass law (continuity equation).

\subsection{Notation}
Let us fix an $n$-dimensional smooth manifold $M$. In local computations we will take coordinates $x^1,\dots,x^n$. Each smooth function $f\in\mathcal{C}^\infty(M)$ defines on the tangent bundle $TM$ the function $\dot f$ by the rule
$$\dot f(v_x):=df(v_x)=v_x(f),\quad\forall  v_x\in T_xM.$$
The set $x^j,\dot x^j,$ gives us a local chart on $TM$ as usual and symbols $\partial_j$ and $\partial_{\dot x^j}$ will denote the respective partial derivatives; in particular, $\dot f=\dot d(f)$ where, by definition, $\dot d:=\dot x^j\partial_j$ (see \cite{MecanicaMunoz} for a intrinsic analysis of this operator). If $u$ is a tangent vector field and $\omega$ a differential form, $\idiota_{\dis u}\omega$ will denote the inner contraction of $\omega$ with $u$. Moreover, when $M$ is endowed with a pseudoriemannian metric tensor $g$ (symmetric 2-covariant tensor without kernel), with each vector field $u$, we associate the differential 1-form $u^\bemol:=\idiota_{\dis u}g$.

\section{Infinitesimal variations on the bundle of fluxes}\label{s:prolongacion}

\subsection{Prolongations of vector fields to the tangent bundle}

Let $M$ be an $n$-di\-men\-sio\-nal manifold.  Its tangent bundle bundle $TM$ is endowed with the so called \emph{contact system} $\Omega$ which is comprised by the differential 1-forms vanishing on the prolongation of parameterized curves: given a differentiable curve $t\to\gamma(t)\in M$, we can consider its prolongation $\widehat\gamma\colon t\mapsto (d_t\gamma)(d/dt)_t\in T_{\gamma(t)}M$. In local coordinates, if $\gamma$ is described by $x^j=\gamma^j(t)$, then $\widehat\gamma$ is
$$(x^j=\gamma^j(t),\dot x^j=(d\gamma^j(t)/dt))=\frac {d\gamma^j(t)}{dt}(\partial_j)_{\gamma(t)},$$
and so, the contact system $\Omega$ is generated by the 1-forms
$$\dot x^\ell dx^j-\dot x^jdx^\ell,\quad j,\ell=1,\dots,n.$$

On $TM$ there is defined the infinitesimal generator of the homotheties (along fibres) which is the vector field $W$ that have the local expresion
$$W=\dot x^j\partial_{\dot x^j}.$$

For a given tangent field $\xi$ on $M$, a little computation shows that we can ``prolongate'' it to the tangent bundle in several ways by imposing the invariance of the contact system:
\begin{prop}[see (\cite{RM})]\label{prolongations}
Let $W$ be the infinitesimal generator of the homotheties along fibres of $TM$. The set of vector fields on $TM$ that project onto $\xi=\xi^j\partial_j$, and preserve the contact system by Lie derivative, is
$$\xi^j\partial_j+\dot \xi^j\partial_{\dot x^j}+\varphi W$$
where $\varphi$ is an arbitrary function on $TM$.
\end{prop}

 Therefore, it is necessary to impose some additional condition in order to determine a prolongation. For example, we can consider the \emph{class of time} comprised by the horizontal forms, $\alpha=\alpha^j(x,\dot x)dx^j$, such that if we define $\dot d:=\dot x^j\partial_j$ (the ``generic velocity''), then $\idiota_{\dis\dot d}\,\alpha=\alpha^j(x,\dot x)\dot x^j=1$ (the 1-forms $df/\dot f$, $f\in\A$, generate the class of time module contact forms; see \cite{MecanicaMunoz} for details). If now we impose on the possible prolongations of a field $\xi$ the condition of preserving the class of time we get the usual prolongation
$$\xi=\xi^j\partial_j\mapsto \textrm{pr}^{(1)}\xi:=\xi^j\partial_j+\dot \xi^j\partial_{\dot x^j}.$$
This is called \emph{covariant prolongation} of $\xi$, or \emph{vertical lift} of $\xi$ in \cite{Ootsuka1}. However, as we will see (Subsection \ref{subfluxprolongacion}), another prolongations are also natural and useful depending on the context.
\medskip

On the other hand, each vector tangent to $M$ at a point $x$, say $e_x$, defines a tangent vector $(e_x)^{\dis v_x}\in T_{v_x}(T_xM)$, for any $v_x\in T_xM$ due to the  linear structure of such a fibre by using the directional derivative:
 \be\label{e:representantevertical}
 (e_x)^{\dis v_x}f:=\left.\frac{d}{dt}\right|_{t=0}f(v_x+te_x)
 \ee
 for each function $f\in\C^\infty(TM)$.

 In local coordinates, when $e_x=e_x^j(\partial_j)_x$ we have $(e_x)^{\dis v_x}=e_x^j(\partial_{\dot x^j})_{v_x}$.

If, instead, we have a vector field $e$ on $M$ we get a vertical vector field $E$ on $TM$, that we call the \emph{vertical representative} of $e$ (see \cite{MecanicaMunoz}); in local coordinates,
\begin{equation}\label{representantevertical}
e=e^j\partial_j\mapsto E=e^j\partial_{\dot x^j}.
\end{equation}

If now we have two vector fields $e$ and $v$ on $M$, we can get a tangent field $e^{\dis v}$ defined on the submanifold $v(M):=\{v_x\,|\,x\in M\}\subset TM$ by the rule:
$$e^{\dis v}\colon v_x\mapsto (e_x)^{\dis v_x}\in T_{v_x}T_xM.$$
In other words, $e^{\dis v}$ is the restriction of the vertical representative $E$ to the submanifold $v(M)$.

\subsection{Lifting of tangent fields to the bundle of fluxes}
When we have a vector field $u$  and a volume form $\vol$, the integral of the  $(n-1)$-form $\idiota_{\dis u}\vol$ on the hypersurface enclosing an open region, measures the flux of $u$ through that hypersurface. This is why we call  \emph{bundle of fluxes} the bundle $\Lambda^{n-1}M\to M$, where $\Lambda^{n-1}M$ denotes the space of differential  $(n-1)$-forms on $M$. A coordinate system  $x^j$ on $M$ and the choice of a (local) volume form, say $\vol$, induces a local chart $x^j,w^j,$ on $\Lambda^{n-1}M$ by the rule
 $$ \sigma_x=w^j(\sigma_x) \idiota_{{\dis \partial_j}}\vol,\quad\text{for each $\sigma_x\in \Lambda^{n-1}M$}.$$

 In the same way in which Liouville form is defined on the cotangent bundle, it is defined on $\Lambda^{n-1}M$ the tautological $(n-1)$-form $\Upsilon$: given an $(n-1)$-form $\sigma_x$ at the point $x$, and tangent vectors $D^k_{\sigma_x}\in T_{\sigma_x}\Lambda^{n-1}M$, $k=1,\dots,n-1$, we set
 $$\Upsilon_{\sigma_x}(D^1_{\sigma_x},\dots,D^{n-1}_{\sigma_x}):=\sigma_x(\pi_*D^1_{\sigma_x},\dots,\pi_*D^{n-1}_{\sigma_x}),$$
 where $\pi_*$ is the tangent map corresponding to the projection $\pi\colon\Lambda^{n-1}M\to M$.
 In local coordinates, it is described by
$$\Upsilon=w^j\, \idiota_{{\dis \partial_j}}\vol=\idiota_{{\dis w^j\partial j}}\vol,$$
where we assume that $n$-form $\textrm{vol}$ of $M$ is pull-backed to $\Lambda^{n-1}M$ by means of $\pi$.
A computation shows the following
\begin{lema}\label{subidaalosflujos}
Each vector field $\xi$ on $M$ determines another one, $\widetilde \xi$, on $\Lambda^{n-1}M$ which is the unique one which projects onto $\xi$ and preserves the tautological form $\Upsilon$; that is to say,
$$\mathcal L_{\dis\widetilde \xi}\Upsilon=0$$
($\mathcal L_{\dis\widetilde\xi}$ denotes de Lie derivative operator).
\end{lema}
\begin{proof}
Let $\xi=\xi^j\partial_j$ in local coordinates; so, $\widetilde \xi=\xi^j\partial_j+B^j\partial_{w^j}$ for some $B^j$ to be determined. We have,
\begin{align*}
\mathcal L_{\dis\widetilde \xi}\,\Upsilon &=\mathcal L_{\dis\widetilde \xi}\,(\idiota_{\dis w^j\partial_j}\vol) \\
                 &=\idiota_{\dis \mathcal L_{\dis \widetilde{\xi}}\,(w^j\partial_j)}\vol+\idiota_{\dis w^j\partial_j}\mathcal L_{\dis \widetilde{\xi}}\,\vol    \\
                 &=\idiota_{\dis \{B^j\partial_j-w^h\partial_h(\xi^j)\partial_j-\partial_h(B^j)\partial_{w^j}\}}\vol+\idiota_{\dis w^j\partial_j}\textrm{div}(\xi)\vol\\
                 &=\idiota_{\dis \{(B^j-w^h\partial_h(\xi^j)+\diver(\xi)\,w^j)\partial_j\}}\vol
\end{align*}
where we have used that $\idiota_{\partial_{w^j}}\vol=0$ and we denote by $\diver(\xi)$  the divergence of $\xi$ with respect to $\vol$: $\mathcal L_{\dis \xi}\vol=\diver(\xi)\,\vol$. Thus, $\mathcal L_{\dis \widetilde \xi}\Upsilon$ vanishes if and only if $B^j=w^h\partial_h(\xi^j)-(\diver(\xi))w^j$. As a consequence, each $B^j$ is completely determined and the uniquely defined $\widetilde \xi$
is, in local coordinates,
$$\widetilde \xi=\xi^j\partial_j+\left(w^h\partial_h(\xi^j)-\textrm{div}(\xi)w^j\right)\partial_{w^j}.$$
\end{proof}

\subsection{Transfer to the tangent bundle}\label{subfluxprolongacion}
Let us assume that $M$ is oriented by a global volume form $\vol$. For instance, if $M$ is orientable and endowed with a pseudoriemannian metric $g$ we can take the volume form $\vol$ associated with $g$; locally,  $\vol:=\sqrt{}{\phantom{.}}dx^1\wedge\cdots\wedge dx^n$, where $\sqrt{}$ denotes the squared root of $|\textrm{det}g|$.

In this case, we get the isomorphism
$$\phi\colon TM\to\Lambda^{n-1}M,\quad u_x\mapsto \phi(u_x):=\idiota_{\dis u_x}\textrm{vol},$$
which in the above introduced coordinates can be written simply as $w^j=\dot x^j$. In this way, $\widetilde \xi$ defined on the bundle of fluxes is transferred to the tangent bundle as the vector field (we keep notation)
$$\widetilde \xi=\xi^j\partial_j+\left(\dot x^h\partial_h(\xi^j)-\diver(\xi)\dot x^j\right)\partial_{\dot x^j}=\xi^j\partial_j+\dot \xi^j\partial_{\dot x^j}-\textrm{div}(\xi)\, W,$$
and then we get the prolongation
\be\label{fluxprolongacion}
\xi\mapsto \widetilde \xi=pr^{(1)}\xi-\diver(\xi)W,
\ee
which belongs to the set of contact preserving prolongations of $\xi$ according to Proposition \ref{prolongations}. We will call $\widetilde \xi$ the \emph{flux prolongation} of $\xi$ with respect to the pseudometric $g$. This prolongation is the most appropriate if what we want is to look at the flux as the main property of a vector field.

\begin{obs}
A prolongation of this type is used in \cite{Anton} to formulate fluids in the context of a constrained variational approach.
\end{obs}

\bigskip

\section{Flux variational problems}\label{s:problemas variacionales}
We keep the previous notation and hypotheses.
We will consider variational problems defined on the tangent bundle with respect to variations defined according the flux prolongations. For a given
$ L\in\cC^\infty(TM)$ let us define the functional:
$$u\overset{I}\longmapsto I(u):=\int_{u(M)}\lambda,$$
where $u\colon M\to TM$ is a tangent vector field and $\lambda= L\,\vol$  is a lagrangian density on $TM$.

Now, we apply the variations of $I$ at $u$ associated with compact supported vector fields $\xi$ by means of $\widetilde \xi$ (see (\ref{fluxprolongacion})):
\begin{equation}\label{variation}
\delta_{\dis \xi}I(u):=\int_{u(M)}\mathcal L_{\dis \widetilde \xi}\lambda=\int_M u^*\mathcal L_{\dis \widetilde \xi}\lambda.
\end{equation}

The pull-back $u^*\mathcal L_{\dis \widetilde \xi}\lambda$ depends only on the values of $\widetilde\xi$ when restricted to the section $u(M)\subseteq TM$; these values are   
$$
\widetilde \xi_{u_x}=\left(\xi^j\partial_j+\dot \xi^j\partial_{\dot x^j}-\textrm{div}(\xi)\, W\right)_{u_x}
  =\xi^j(x)\partial_{j,u_x}+u_x(\xi^j)\partial_{\dot x^j,u_x}-\textrm{div}_x(\xi)\, u_x^{\dis u_x} 
$$
where $u_x^{\dis u_x}$ is defined as in (\ref{e:representantevertical}).  In addition, taking into account that, for each point $x\in M$,
  $$(u_*\xi)_x=(\xi^j\partial_j)_{u_x}+\xi_x(u^j)(\partial_{\dot x^j})_{u_x},$$
  we arrive to
  \be\label{xiprimera}
  \widetilde \xi_{u_x}=u_*\xi_x-\left([\xi,u]_x+\textrm{div}_x(\xi)\,u_x\right)^{\dis u_x}
  \ee
where, for each tangent vector $v_x\in T_xM$, $v_x^{\dis u_x}$ is defined as in (\ref{e:representantevertical}).

Let us put by definition
$$\delta_{\dis \xi}u:=[\xi,u]+\textrm{div}(\xi)\,u,$$
in such a way that (\ref{xiprimera}) is written as
 $$
  \widetilde \xi_{u_x}=u_*\xi_x-(\delta_\xi u)_x^{\dis u_x}
 $$

 In this notation we states the following
 \begin{prop}
 For each vector field $\xi$ with compact support, the (flux) variation of $I$ at the section $u$ is
 $$\delta_{\dis \xi}I(u)=-\int_M(\delta_{\dis \xi}u)^{\dis u}(L)\,\vol,$$
 where $\delta_{\dis \xi}u:=[\xi,u]+\textrm{div}(\xi)\,u,$ and, by its very definition (\ref{e:representantevertical}), the derivative in above integrand is
\be\label{evariation1}
(\delta_{\dis \xi}u)^{\dis u}(L)=\left(\left.\frac{d}{dt}\right|_{t=0}L(u+t\delta_{\dis \xi}u)\right).
\ee
 \end{prop}
 \begin{proof}
 Let $\gamma_1$ (respectively, $\gamma_2$) a tangent  field on $TM$ that coincides on $u(M)$ with the field $u_x\mapsto u_*\xi_x$ (respectively, $u_x\mapsto(\delta_\xi u)_x^{\dis u_x}$), which is defined only for the points $u_x\in u(M)$. Further, we can assume that $\gamma_2$ is vertical with respect to $TM\to M$. Then,
 \be\label{xisegunda}
 u^*\mathcal L_{\dis \widetilde \xi}\lambda=u^*\mathcal L_{\dis\gamma_1}\lambda-u^*\mathcal L_{\dis\gamma_2}\lambda,
 \ee
 because $\widetilde\xi$ and $\gamma_1-\gamma_2$ are equal on $u(M)$.
 Now, by the Cartan formula for the Lie derivative and the identity $u^*d\lambda=du^*\lambda=0$ (since an $(n+1)$-differential form on the $n$-dimensional manifold $M$ must vanish), it holds
 \be\label{xitercera}
 u^*\mathcal L_{\dis\gamma_1}\lambda=d\idiota_\xi u^*\lambda.
 \ee
 Similar arguments, but now taking into account the verticality of $\gamma_2$ and the horizontality of $\lambda$, imply that
 \be\label{xicuarta}
 u^*\mathcal L_{\dis\gamma_2}\lambda=u^*\idiota_{\dis\gamma_2}d\lambda=u^*\gamma_2(L)\,\vol=(\delta_\xi u)^u(L)\,\vol.
 \ee
 From (\ref{xisegunda}), (\ref{xitercera}), (\ref{xicuarta}),
 $$u^*\mathcal L_{\dis \widetilde \xi}\lambda=d\idiota_\xi u^*\lambda-(\delta_\xi u)^u(L)\,\vol.$$
 If, besides that, $\xi$ is compact supported, the term $d\idiota_\xi u^*\lambda$ disappears when is integrated. Thus the proof is finished.
 \end{proof}

\subsection{Variation of mass density}
In this section we assume that $M$ is orientable and endowed with a pseudoriemannian metric $g$ and $\vol$ is associated with $g$; so that $\vol=\sqrt{}{\phantom{.}}dx^1\wedge\cdots\wedge dx^n$ locally.

Physical considerations (see \cite{DiracRelatividad}, p. 50)  lead to treat the ``density'' function in the tangent bundle:
 $$\rho\colon TM\to\mathbb{R},\qquad \rho(u_x)=\sqrt{|g(u_x,u_x)|};$$
 that is to say, $\rho$ 
 is the length function.
  If $ L=\Phi(\rho)$ (for a certain function $\Phi)$,  then
\be
(\delta_{\dis \xi}u)^{\dis u}\Phi(\rho)=\pm\Phi'(\rho_u)\,\frac{g(u,\delta_{\dis \xi}u)}{\sqrt{|g(u,u)|}}=\pm\Phi'(\rho_u)\,g(v,\delta_{\dis \xi}u),
\ee\label{evariation2}
where, $\rho_u=\sqrt{|g(u,u)|}$, $v:=u/\rho_u$, the sign $\pm$ is that of $g(u,u)$ and we have applied (\ref{evariation1}) and
\begin{align*}
(\delta_{\dis \xi}u)^{\dis u}\rho &=\frac 1{2\sqrt{\rho_u^2}}(\delta_{\dis \xi}u)^{\dis u}\rho^2=
     \frac 1{2\sqrt{\rho_u^2}}\left(\left.\frac{d}{dt}\right|_{t=0}\rho^2(u+t\delta_{\dis \xi}u)\right)\\[5pt]
   &=\frac {1}{2\rho_u}\left(\pm\left.\frac{d}{dt}\right|_{t=0}g(u+t\delta_{\dis \xi}u,u+t\delta_{\dis \xi}u)\right)=
    \pm\frac{2\,g(u,\delta_{\dis \xi}u)}{2\rho_u}=\pm g(v,\delta_{\dis \xi}u).
\end{align*}

As a consequence,
\begin{equation}\label{Calculovariacion}
\delta_{\dis \xi}I(u)=\mp\int_M\Phi'(\rho_u )\,g(v,\delta_{\dis \xi}u)\,\vol.
  \end{equation}
\medskip

Next, we need the following
\begin{lema}\label{identidad2}
For each couple of vector fields $p,q$ on $M$, if we define $p^\bemol:={\emph \idiota}_{\dis p}g$,
it holds
$$p^\bemol\wedge {\emph\idiota}_{\dis q}\vol=p^\bemol(q)\,\vol=g(p,q)\,\vol.$$
\end{lema}
\begin{proof}
By taking the inner product of $q$ by the null $(n+1)$-form $p^\bemol\wedge\vol$ we get
$$0=\idiota_{\dis q}(p^\bemol\wedge\vol)=(\idiota_{\dis q}p^\bemol)\vol-p^\bemol\wedge\idiota_{\dis q}\vol.$$
Since  $\idiota_{\dis q}p^\bemol=p^\bemol(q)=g(p,q)$, the proof is finished.
\end{proof}
\bigskip

By the above lema and the identity $$\idiota_{\dis \delta_{\dis \xi}u}\vol=\idiota_{\dis [\xi,u]+(\textrm{div}\,\xi) u}\vol=\mathcal L_{\dis \xi}(\idiota_{\dis u}\vol),$$
we arrive to
$$g(v,\delta_{\dis \xi}u)\,\vol=v^\bemol\wedge \idiota_{\dis \delta_{\dis \xi}u}\,\vol=v^\bemol\wedge \mathcal L_{\dis \xi}(\idiota_{\dis u}\vol).$$
With this expression, the integrand of (\ref{Calculovariacion}) becomes
\begin{equation}\label{desarrollo1}
\Phi'(\rho_u )\,v^\bemol\wedge \mathcal L_{\dis \xi}(\idiota_{\dis u}\vol)=
 \mathcal L_{\dis \xi}\left(\Phi'(\rho_u )\,v^\bemol\wedge (\idiota_{\dis u}\vol)\right)-\mathcal L_{\dis \xi}(\Phi'(\rho_u )\,v^\bemol)\wedge \idiota_{\dis u}\vol.
 \end{equation}
The first summand in the right hand side of (\ref{desarrollo1}) is an exact differential form because is the Lie derivative of a differential $n$-form on the $n$-dimensional manifold $M$; so, we can discard it if we integrate when the field $\xi$ has compact support. Let us denote with the symbol $\equiv$ the equality up to terms of the above type (those that can be discarded in the integration):
$$\Phi'(\rho_u )\,v^\bemol\wedge \mathcal L_{\dis \xi}(\idiota_{\dis u}\vol)\equiv-\mathcal L_{\dis \xi}(\Phi'(\rho_u )\,v^\bemol)\wedge \idiota_{\dis u}\vol$$
that, by the Cartan formula, is (minus) the sum of
$$\alpha :=d(\idiota_{\dis \xi}\Phi'(\rho_u )\,v^\bemol)\wedge \idiota_{\dis u}\vol\quad\text{and}\quad
\beta :=\idiota_{\dis \xi}d(\Phi'(\rho_u )\,v^\bemol)\wedge \idiota_{\dis u}\vol,
$$
but
$$\alpha=d\left(\idiota_{\dis \xi}\Phi'(\rho_u )v^\bemol\wedge \idiota_{\dis u}\vol\right)-(\idiota_{\dis \xi}\Phi'(\rho_u )v^\bemol)d\idiota_{\dis u}\vol\equiv
-(\idiota_{\dis \xi}\Phi'(\rho_u )v^\bemol)\textrm{div}(u)\vol.$$
In addition, for any 1-form $\sigma$ and vector field $q$ we have (similarly to Lema \ref{identidad2})
$$0=\idiota_q(\sigma\wedge\vol)=(\idiota_q\sigma)\vol-\sigma\wedge\idiota_q\vol$$ so that
$$\beta=\idiota_{\dis u}\idiota_{\dis \xi}d(\Phi'(\rho_u )v^\bemol)\vol,$$
and then,
\begin{equation}\label{desarrollo2}
\Phi'(\rho_u )\,v^\bemol\wedge \mathcal L_{\dis \xi}(\idiota_{\dis u}\vol)\equiv \idiota_{\dis \xi}\left\{\Phi'(\rho_u )\,\textrm{div}(u)v^\bemol+\idiota_{\dis u}d(\Phi'(\rho_u )v^\bemol)\right\}\,\vol.
\end{equation}

We derive the
\begin{prop}\label{propvariacion}
In the above notation
$$\delta_{\dis \xi}I(u)=0,\quad\text{for all $\xi$ with compact support,}$$
 if and only if
\begin{equation}\label{seccionescriticas}
\Phi'(\rho_u )\,\textrm{div}(u)v^\bemol+\emph\idiota_{\dis u}d(\Phi'(\rho_u )v^\bemol)=0
\end{equation}
where $u=\rho_uv$.
\end{prop}
\begin{proof}
Acording to the result established in Equation (\ref{desarrollo2}) it holds that
$$\delta_{\dis \xi}I(u)=\mp\int_M\idiota_{\dis \xi}\left\{\Phi'(\rho_u )\,\textrm{div}(u)v^\bemol+\idiota_{\dis u}d(\Phi'(\rho_u )v^\bemol)\right\}\,\vol.$$
If we assume that $\delta_{\dis \xi}I(u)$ vanishes for arbitrary $\xi$ of compact support we get the statement.
\end{proof}
\bigskip

\subsection{The fluid equations}

Let us see subsequently, the simplest non trivial case and then, the general one.
\medskip

\noindent {\bf Case $\Phi(\rho)=\rho$:} Now $\Phi'=1$ and thus (\ref{seccionescriticas}) becomes
\begin{equation}\label{critica1}
 \textrm{div}(u)v^\bemol+\idiota_{\dis u}dv^\bemol=0.
 \end{equation}
By contracting with $u$ and taking into account that $\idiota_{\dis u}\idiota_{\dis u}dv^\bemol=0$ and $\idiota_{\dis u}v^\bemol=g(u,v)=\pm\rho_u$ (according to the sign of $g(u,u)$); it follows that
$$\textrm{div}(u)\rho_u=0\quad \Rightarrow\quad \textrm{div}(u)=0;$$
that is to say,  $u$ is conservative.
By substituting into (\ref{critica1}), we see that $\idiota_{\dis u}dv^\bemol=0$. But $u=\rho_uv$, and then, that is equivalent to
\begin{equation}\label{identidad1}
\idiota_{\dis v}dv^\bemol=0
\end{equation}
(at least where $\rho_u\ne 0$).

Next, we will need to apply the following
\begin{lema}\label{identidad}
For each vector field $w$, it holds the identity
$$\emph\idiota_{\dis w}dw^\bemol=(w^\nabla w)^\bemol-\frac 12\,d(g(w,w))$$
(where $\nabla$ denotes the Levi-Civita connection associated with $g$)
\end{lema}
\begin{proof}
(\cite{MecanicaMunoz}, Remark 1.6, \cite{RM}, Section 0). Let $z$ be another vector field. Then
\begin{align*}
(\idiota_wdw^\bemol)(z) &=dw^\bemol(w,z)\overset{(1)}=w(w^\bemol(z))-z(w^\bemol(w))-w^\bemol([w,z])\\
                &\overset{(2)}=w(g(w,z))-z(g(w,w))-g(w,[w,z])\\
                &\overset{(3)}=g(w^\nabla w,z)+g(w,w^\nabla z)-z(g(w,w))-g(w,[w,z])
\end{align*}
where we have used: (1) the definition of exterior differential, (2) the definition of $w^\bemol$ and (3) one of the characterizing properties of the Levi-Civita connection. Moreover, this connection is symmetric so that, $[w,z]=w^\nabla z-z^\nabla w$, and then
$$(\idiota_wdw^\bemol)(z)=g(w^\nabla w,z)+g(w,z^\nabla w)-z(g(w,w));$$
If, in addition, we take into account that
$$g(w,z^\nabla w)=\frac 12 z(g(w,w))=d(g(w,w)/2)(z)$$
and substitute, we arrive to
$$(\idiota_wdw^\bemol)(z)=g(w^\nabla w,z)-d(g(w,w)/2)(z)=\left((w^\nabla w)^\bemol-d(g(w,w)/2)\right)(z),$$
as required.
\end{proof}
From Lema \ref{identidad} and Equation (\ref{identidad1}), by taking into account that $g(v,v)=\pm 1$, we arrive to
\begin{equation}\label{geodesico}
v^\nabla v=0.
\end{equation}

Putting all together,
\begin{prop}
The critical sections for the functional
$$I(u)=\int_u\rho\,\vol$$
under flux variations
are the conservative fields $u$ such that $v=u/\rho_u$ (the unitary field associated with $u$) is a  geodesic field.
\end{prop}
\bigskip

\noindent {\bf General case $\Phi(\rho)$:}
Exactly as in the above case, by inner contraction of (\ref{seccionescriticas}) with $u$ it is derived that $u$ is conservative: $\textrm{div}(u)=0$ (at least in the open set where $\Phi'(\rho_u )\rho_u\ne 0$). Then, by substituting into (\ref{seccionescriticas}),
$$\idiota_{\dis u}d(\Phi'(\rho_u )v^\bemol)=0.$$
being $u=\rho_u\, v$ with $\sqrt{|g(v,v)|}=1$, we have
\begin{equation}\label{critica2}
0=\idiota_vd(\Phi'(\rho_u)v^\bemol)=\idiota_v(d\Phi'\wedge v^\bemol+\Phi'dv^\bemol)=v(\Phi')v^\bemol\mp d\Phi'+\Phi'(v^\nabla v)^\bemol,
\end{equation}
where we have applied that $\idiota_vv^\bemol=g(v,v)=\pm1$ and, then, that by Lemma \ref{identidad}, $i_vdv^\bemol=(v^\nabla v)^\bemol$.

Now, by passing to the dual using de metric we have $d\Phi'\mapsto\textrm{grad}\,\Phi'$, $v^\bemol\mapsto v$, etc.,  in such a way that  Equation (\ref{critica2}) becomes
\begin{equation}\label{critica3}
0=v(\Phi')v\mp\textrm{grad}\,\Phi'+\Phi'\,v^\nabla v
\end{equation}

In order to put Equation (\ref{critica3}) in a more usual form we take the following definitions. Firstly, let $\epsilon$ the function such that $\Phi(\rho)=\rho(1+\epsilon(\rho))$.
Since $\Phi'=1+\epsilon+\rho\epsilon'=\frac{\Phi+\rho^2\epsilon'}\rho$ it results, by denoting $P:=\rho^2\epsilon'$, that
$\Phi'=(\Phi+P)/\rho$ and that
$$\Phi''=\frac{(\Phi'+P')\rho-(\Phi+P)}{\rho^2}=\frac{(\Phi+P)+P'\rho-(\Phi+P)}{\rho^2}=\frac{P'}\rho.$$

In this way, $v(\Phi')=\Phi'' v(\rho)=(P'/\rho)v(\rho)=v(P)/\rho$ and analogously $\textrm{grad}(\Phi')=\textrm{grad}(P)/\rho$. Therefore,

\begin{prop}
The critical sections for the functional
$$I(u)=\int_u\Phi(\rho)\,\vol$$
under flux variations
are the \emph{conservative fields} $u$ ($\diver\, u=0$, \emph{continuity equation}) such that
\smallskip

\begin{equation*}
0=(\Phi+P)v^\nabla v+v(P)v\mp\textrm{grad}\,P,
\end{equation*}
\smallskip

\noindent where $\Phi=\rho_u\,(1+\epsilon(\rho_u))$, $P:=\rho_u^2\epsilon'(\rho_u)$, $u=\rho_u v$, $g(v,v)=\pm1$.
\end{prop}
\bigskip

In the particular case of a couple $(M,g)$ being a Lorentzian manifold, these are the continuity and the \emph{Euler equations for a relativistic perfect fluid}.

\subsection*{Conclusions}
Based in physical reasons, we have considered the natural variational problem consisting of functionals of the length of vector fields. Without the necessity of impose any kind of constraints nor potentials, we get the well known equations of a relativistic perfect fluid consisting of the Euler equations joint with the conservation of mass law (indeed, more than that is obtained because we get analogous equations on any pseudoriemannian oriented manifold). The only price to pay is to take variations according the flux produced by vector fields. On the other hand, we think this is a quite natural and justified procedure.


\end{document}